\documentclass[12pt,a4paper,abstract=true]{scrartcl} 
\usepackage[a4paper,margin=2.5cm,includefoot
]{geometry}
\usepackage[T1]{fontenc}
\usepackage[utf8]{inputenc}
\usepackage[english]{babel}
\usepackage{csquotes}
\usepackage{lmodern}
\usepackage{amssymb}
\usepackage{amsmath}
\usepackage{amsthm}
\usepackage{enumitem}
\usepackage{mathtools}
\usepackage{mathrsfs}
\usepackage{dsfont}
\usepackage{color}
\usepackage{braket}

\usepackage[style=alphabetic,giveninits=true,maxnames=5, backend=biber,doi=false,url=false,isbn=false,eprint=true]{biblatex}
\PassOptionsToPackage{hyphens}{url}

\renewcommand{\:}{\colon}

\usepackage[breaklinks]{hyperref}
\usepackage{cleveref}

\newcommand{\IR}{{\mathbb{R}}}
\newcommand{\IN}{{\mathbb{N}}}
\newcommand{\IZ}{{\mathbb{Z}}}
\newcommand{\IP}{{\mathbb{P}}}

\newcommand{\IC}{{\mathbb{C}}}

\newcommand{\abs}[1]{\left| #1 \right|}

\renewcommand{\d}{\,\mathrm{d}}

\renewcommand{\sc}[1]{\left< #1 \right>}
\newcommand{\nn}[1]{\left\Vert #1 \right\Vert}

\newcommand{\ind}{1\hspace{-0.8ex}1}

\newcommand{\tr}{\operatorname{tr}}

\newcommand{\supp}{\operatorname{supp}}

\newcommand{\Id}{\operatorname{Id}}

\newtheoremstyle{fgrs} 
                        {0.5em}    
                        {0.5em}    
                        {}         
                        {}         
                        {\bfseries}
                        {}        
                        {\newline} 
                        {}
\newtheoremstyle{mydef} 
                        {0.5em}    
                        {0.5em}    
                        {}         
                        {}         
                        {\bfseries}
                        {}        
                        {\newline} 
                        {}
\theoremstyle{mydef}

\AfterEndEnvironment{defn}{\noindent\ignorespaces}
\AfterEndEnvironment{example}{\noindent\ignorespaces}
\AfterEndEnvironment{rem}{\noindent\ignorespaces}
\AfterEndEnvironment{prop}{\noindent\ignorespaces}
\AfterEndEnvironment{thm}{\noindent\ignorespaces}
\AfterEndEnvironment{lemma}{\noindent\ignorespaces}
\AfterEndEnvironment{coro}{\noindent\ignorespaces}
\AfterEndEnvironment{hyp}{\noindent\ignorespaces}
\AfterEndEnvironment{proof}{\noindent\ignorespaces}

\renewcommand{\H}{\mathcal{H}}

\renewcommand{\d}{{\mathrm d}}

\renewcommand{\Re}{\operatorname{Re}}

\renewcommand{\i}{\mathrm{i}}

\newcommand{\IS}{\mathbb{S}}

\newcommand{\R}{\mathbb{R}}

\newcommand{\Def}{\mathcal{D}}

\newcommand{\q}{\mathsf q}
\newcommand{\z}{z}

\newcommand{\SL}{\operatorname{SL}}
\newcommand{\SO}{\operatorname{SO}}
\newcommand{\PSL}{\operatorname{PSL}}
\newcommand{\Cda}{C_d(a)}

\newcommand{\refJ}[1]{\hyperref[Js]{($\text{J}_{#1}$)}}

\setlist[description]{font=\normalfont\itshape}

\renewcommand{\q}{\hat{\mathsf x}}

\DeclareFieldFormat{postnote}{#1}
\DeclareFieldFormat{multipostnote}{#1}

\newcommand{\EN}{E_N}

\newcommand{\Pwg}{\IP^{\operatorname{WP}}_g}
\newcommand{\Mg}{\mathcal{M}_g }
\newcommand{\cgs}{\psi_0}
\newcommand{\Iso}{\operatorname{Iso}}
\newcommand{\Y}{Y}

\newcommand{\ben}{\begin{displaymath}}
\newcommand{\een}{\end{displaymath}}
\newcommand{\beqn}{\begin{equation}}
\newcommand{\eeqn}{\end{equation}}
\newcommand{\beqna}{\begin{eqnarray*}}
\newcommand{\eeqna}{\end{eqnarray*}}

\numberwithin{equation}{section}
\newtheorem{thm}{Theorem}[section]
\newtheorem{lemma}[thm]{Lemma}
\newtheorem{prop}[thm]{Proposition}
\newtheorem{coro}[thm]{Corollary}

\theoremstyle{definition}

\theoremstyle{remark}
\newtheorem{remark}[thm]{Remark}

\crefname{thm}{Theorem}{Theorems}
\crefname{coro}{Corollary}{Corollaries}
\crefname{hyp}{Hypothesis}{Hypotheses}
\Crefname{hyp}{Hypothesis}{Hypotheses}
\crefname{lemma}{Lemma}{Lemmas}
\Crefname{lemma}{Lemma}{Lemmas}
\crefname{prop}{Proposition}{Propositions}
\crefname{enumi}{}{}
\Crefname{enumi}{}{}
\creflabelformat{enumi}{#2(#1)#3}
\crefname{equation}{}{}
\Crefname{equation}{}{}

  \usepackage{authblk}
  \usepackage{bbm}
\newcommand{\Ee}{\mathcal{E}}

\newcommand{\IH}{\mathbb{H}}

\newcommand{\absh}[1]{d(o,#1)}
\newcommand{\vol}[1]{\operatorname{
vol}(#1)}

\newcommand{\name}[1]{#1}

\allowdisplaybreaks
\begin{document}
\title{Bose-Einstein condensation on hyperbolic spaces}
\author{Marius Lemm}
\author{Oliver Siebert}
\affil{\small{Department of Mathematics, University of Tübingen, Auf der Morgenstelle 10, 72076 Tübingen, Germany}}
\date{February 12, 2022}
\maketitle
\begin{abstract} 
 A well-known conjecture in mathematical physics asserts that the interacting Bose gas exhibits Bose-Einstein condensation (BEC) in the thermodynamic limit. We consider the Bose gas on certain hyperbolic spaces. In this setting, one obtains a short proof of BEC in the infinite-volume limit from the existence of a volume-independent spectral gap of the Laplacian.
\end{abstract}
\section{Introduction}

The Bose gas plays a central role in quantum many-body physics. The key effect it displays is \textit{Bose-Einstein condensation (BEC)}, macroscopic occupation of a single-particle quantum state. BEC is immensely useful for technological applications because it amplifies microsopic quantum effects to macroscopic scales. The original theoretical prediction of BEC was made for an ideal (i.e., non-interacting) Bose gas by Bose and Einstein \cite{bose1924,einstein1924}. A more realistic description of the Bose gas involves interactions between particles, resulting in a full quantum many-body problem. 

To describe $N$ bosons confined to an Euclidean torus $\mathbb T_L=(-L/2,L/2)^3$ and interacting via a repulsive two-particle interaction $V\geq 0$, one uses the Hamiltonian
\begin{equation}
H_N=\sum_{i=1}^N (-\Delta_{x_i}) +\sum_{1\leq i<j\leq N} V(x_i-x_j).
\end{equation}
A famous conjecture, mathematically formulated by Lieb in 1998 \cite{bec_openproblems} but probably several decades older, states that the Hamiltonian $H_N$ exhibits Bose-Einstein condensation in its ground state in the thermodynamic limit. More precisely, Lieb's conjecture asserts that there exists a constant $c>0$ such that
\begin{equation}\label{eq:Liebconjecture}
\int_{\mathbb T_L}\int_{\mathbb T_L}\gamma(x,y)\mathrm{d}x\, \mathrm{d}y \stackrel?\geq c L^3
\end{equation}
where 
$$
\gamma(x,y)=\int_{\mathbb T_L^{N-1}} \Psi_0(x,\mathbf{x}) \Psi_0(y,\mathbf{x})\mathrm{d}\mathbf{x}
$$
denotes the $1$-particle correlation function of the unique ground state $\Psi_0\geq 0$ of $H_N$. Indeed, \eqref{eq:Liebconjecture} captures the macroscopic occupation of a single one-particle state. Note for example that \eqref{eq:Liebconjecture} is satisfied for the non-interacting Bose gas with $V\equiv 0$. 

Despite the central role that the Bose gas plays in mathematical physics, this conjecture remains open. Existing proofs of BEC in the Euclidean thermodynamic limit require reflection positivity \cite{aizenman2004bose} or other special positivity properties \cite{koma2021bose}. Nonetheless, spectacular progress has been achieved on the mathematics of the Bose gas in the past twenty years. For instance, the famous Lee-Huang-Yang formula for the subleading energy correction in the dilute limit $\rho a^3\to 0$ has been rigorously derived \cite{yy,fournais_solovej,fournais_solovej2}, which resolved a longstanding open problem. Many works have contributed to precise understanding of BEC and the ground state energy in the Gross-Pitaevskii scaling regime; see e.g., \cite{dyson,gs_energy, lsy,lieb_seiringer_firstproof,rotating1,greenbook,quantum_finetti,bbcs1,bbcs3,bbcs2,deuchert2020gross,nnt,hainzl}. For further background and references, see \cite{greenbook} and the recent survey \cite{rougerie}.

\subsection{Main result: BEC in the infinite-volume limit on hyperbolic space}
In this paper, we study the Bose gas in a \textit{hyperbolic geometry}. The intriguing features of hyperbolic geometry have long inspired mathematicians and physicists alike.  The first uses of hyperbolic geometry in condensed matter theory were to our knowledge based on the AdS-CFT correspondence principle \cite{Witten1998AntideSS,maldacena1999large} which has deepened our understanding of quantum entanglement and may hold the key to quantum error correction. 

More recently, experimentalists have been able to \textit{physically construct hyperbolic structures in the laboratory} by confining particles to discrete hyperbolic lattices in circuit QED \cite{hu2019quantum,hyp_lattices1,hyp_lattices2, PRXQuantum.2.017003, saa2021higher} and by using topoelectric circuits \cite{lenggenhager2021electric}. These setups can serve, among other things, as tabletop simulators of quantum gravity. These recent experimental advances spawned a new interdisciplinary subfield of theoretical physics that blends condensed-matter physics with quantum information science and general relativity \cite{graphs_to_geometry, maciejko2020hyperbolic, ikeda2021hyperbolic, stegmaier2021universality,zhang2021efimov}.  In particular, it was shown that continuum limits of some of these hyperbolic quantum lattice gases produce suitable hyperbolic continuum models \cite{graphs_to_geometry}. Hyperbolic Bose gases were studied, e.g., in \cite{cognola1993bose,KIRA2015185, zhu2021quantum} for example. 

To summarize, it can be said that hyperbolic geometry has emerged as a viable, if still exotic, theater of condensed-matter physics in general and Bose gases in particular. Nonetheless, quantum many-body physics in hyperbolic space is mathematically considerably less explored than its Euclidean counterpart.

The main result of this paper can be summarized as follows.

\vspace{0.5cm}
\textbf{Main result:} \textit{Interacting Bose gases on two- and three-dimensional hyperbolic manifolds rigorously display Bose-Einstein condensation in the infinite-volume limit.}
\vspace{0.5cm}

One can thus say that the \textit{hyperbolic analog} of the conjecture formulated by Lieb holds true. The change of the geometry from Euclidean to hyperbolic is instrumental to proving our results. The key spectral feature in our appropriately chosen hyperbolic setting is the existence of a \textit{volume-independent spectral gap for the Laplacian}. It will not come as a surprise to experts that the large spectral gap removes many of the central difficulties present in the Euclidean case and leads to a rather short proof of BEC. Of course, one still needs to account for the change to hyperbolic geometry in the analytical arguments, in particular in the relevant two-particle scattering problem, but altogether this can be handled rather straightforwardly. Apart from the size of the spectral gap, the change in geometry then mostly surfaces in a different notion of \textit{scattering length} $a$ that is described in the appendix.

We now describe the models of hyperbolic manifolds that are used in this paper and summarize the result on BEC for these models. The precise setup is discussed in more detail in Section \ref{sec:model mainresults}.

\paragraph{Model 1: Quotients by congruence subgroups.} We consider hyperbolic manifolds of the form
$$
X_L=\mathbb H^d / \Gamma(L),\qquad L\geq 2
$$
where $d\in\{2,3\}$ and  $ \Gamma(L)$ is a group of isometries (congruence subgroup). The $\{X_L\}_{L\geq 2}$ form a family of non-compact hyperbolic manifolds with finite volume increasing to infinity,
$$
\mathrm{vol}(X_L)\to\infty,\qquad L\to\infty.
$$
Since they are generated by quotienting the whole space with respect to isometries, the $X_L$'s can be regarded as natural \textit{hyperbolic analogs of Euclidean tori}. For $d=2$, the $X_L$ are known as modular surfaces.

Let us now formulate the result on BEC. Let $\rho=\frac{N}{\mathrm{vol}(X_L)}$ be the particle density. We fix a potential $V\geq 0$ with $\supp V \subseteq B_{R_0}(0)$ for some finite range $R_0>0$. We write $a$ for a hyperbolic analog of the scattering length; see Appendix \ref{appendix} for its precise definition. Then we introduce the auxiliary parameter
\[
\Y = \begin{cases} \rho \ln \frac{1}{\tanh (a/2)}, &\textnormal{for }d = 2, \\ \rho \tanh a, &\textnormal{for }d=3.  \end{cases}
\]
We prove that
\begin{equation}\label{eq:mainresultXL}
\lim_{\Y \to 0} \lim_{\substack{N,L\to\infty\\ \rho=\frac{N}{\vol{X_L}}}}\sc{ \cgs, \gamma  \cgs}= 1, \qquad \textnormal{where }\psi_0=\frac{\mathbbm 1_{X_L}}{\sqrt{\mathrm{vol}(X_L)}},
\end{equation}
where inner product is taken with respect to $L^2(X_L)$.

Comparing with \eqref{eq:Liebconjecture}, we see that \eqref{eq:mainresultXL} indeed proves BEC in the infinite-volume limit for any sufficiently small $Y$. Crucially, ``sufficiently small'' does not depend on the system size $L$ since the $Y$-limit is taken after the infinite-volume limit in \eqref{eq:mainresultXL}. In fact, \eqref{eq:mainresultXL} proves that the occupation of the $\psi_0$-state converges to $1$ as $Y\to 0$. If all the particles belong to a single state up to subleading errors, one speaks of \textit{complete condensation}.

The precise statements for $d=2$ and $d=3$ are given in \cref{th:coro 2d,th:coro 3d} below. They show that the requirement that $Y$ is ``sufficiently small'' can be made explicit in terms of the other system parameters. Moreover, the same condition on $Y$ being sufficiently small also applies to \textit{finite} systems of the same density. The infinite-volume limit in \eqref{eq:mainresultXL} is added only for emphasis and the actual result is more general.

\paragraph{Model 2: Random compact hyperbolic surfaces.} There is a natural probability measure on compact hyperbolic manifolds of fixed volume called the Weil-Petersson measure $\mathbb P^{\mathrm{WP}}_g$. Let $\mathcal M_g$ denote the set of compact hyperbolic surfaces with genus $g$ up to isometry. By the Gauss-Bonnet theorem, the volume of any hyperbolic surface $X\in\mathcal M_g$ equals $2\pi (2g-2)$. By taking $g\to\infty$, we obtain compact hyperbolic manifolds whose volumes go to infinity.

Let $\varepsilon > 0$. We prove that there exists a family of measurable subsets (events) $\mathcal A_g\subset \mathcal M_g$ such that  
\begin{equation}\label{eq:mainrandom}
 \lim_{\rho \ln\frac{1}{\tanh a}\to 0 }  \lim_{ \substack{N,g\to\infty :\\ \rho=\frac{N}{2\pi(2g-2)}}}\Pwg (\mathcal A_g)=1
\end{equation}
and for every hyperbolic manifold $X\in\mathcal A_g$, we have
\[
\sc{ \cgs, \frac{\gamma}{N}  \cgs}  \geq 1 - \varepsilon, \qquad \textnormal{where }\psi_0=\frac{\mathbbm 1_{X}}{\sqrt{\mathrm{vol}(X)}}.
\]
This proves that BEC occurs for random compact hyperbolic manifolds with probability going to $1$ in the infinite-volume limit. For the precise statement, see \cref{th:coro random}.

These results rely on two main estimates for general hyperbolic Bose gases, Proposition \ref{th:lower bound} and Theorem \ref{th:upper bound coro}. The key common feature of the hyperbolic models described above is that the spectral gap of the corresponding Laplacian (Laplace-Beltrami operator) is bounded from below independently of the volume by deep results of Selberg \cite{selberg2} and Mirzakhani \cite{mirzakhani}. Quantitative improvements followed in \cite{sarnak3d,li1987poincare,arbitrarydimension,burger_sarnak,clozel2003demonstration,kim16functoriality}, respectively \cite{lambda_positive_random1,lambda_positive_random2}.

\subsection{Comparison to the Gross-Pitaevskii regime}\label{ssect:GP}
As mentioned above, a commonly studied scaling limit in Euclidean setting is the \textit{Gross-Pitaevskii (GP) regime} which is suitable for describing dilute Bose gases with strong short-ranged interaction. It is characterized by linking the length scale $L$ of the torus to the particle density $\rho=\frac{N}{L^3}$ and the scattering length  $a\in\R$ (which captures the essential features of strength and range of the potential $V$ for two-boson scattering) via
\begin{equation}\label{eq:GPlength}
L=C\frac{1}{\sqrt{\rho a}}
\end{equation}
Combined with the dilute limit $\rho a^3\to 0$ this defines the GP scaling limit. This is in contrast to the thermodynamic limit (which is the subject of the open conjecture in the Euclidean setting and which we consider here in the hyperbolic setting), where one can take $N,L\to\infty$ independently for  \textit{fixed} values of $\rho$ and $a$.

The investigation of the ground state energy asymptotics and BEC in the GP regime is a success story of mathematical physics. Landmark works in this direction include Dyson's study \cite{dyson}, the various works of Lieb, Seiringer, and Yngvason \cite{gs_energy, lsy,lieb_seiringer_firstproof,rotating1} (see also \cite{greenbook}) and more recent advances \cite{bbcs1,bbcs3,bbcs2,nnt}

 Recent approaches rigorously implement a heuristic description of excitations above the condensate due to Bogoliubov \cite{bogoliubov} through localization in Fock space, higher-order Bogoliubov transformations, and other technical innovations. The occurrence of BEC can be pushed to length scales larger than the GP scale \eqref{eq:GPlength} by further exploiting the close link between energy estimates and BEC on different spatial scales \cite{greenbook, fournais_lengthscales,adhikari2021bose}. For instance, it was shown in \cite{fournais_lengthscales} that BEC occurs up to length scales
$$
L=C\frac{(\rho a^3)^{-\delta}}{\sqrt{\rho a}},\qquad 0<\delta<\frac{1}{4}.
$$
and the upper bound on $\delta$ could be further improved somewhat by using methods from \cite{fournais_solovej} as described in  \cite{fournais_lengthscales}.

Despite these advances, the conjectured occurrence of BEC in the thermodynamic limit (i.e.\ for $L$ arbitrarily large) has remained open. A key reason for the failure of these ``energy methods'' beyond certain length scales is the fact that the spectral gap of the Laplacian on the torus $\mathbb T_L=(-L/2,L/2)^3$ vanishes as $L^{-2}$ for $L\to\infty$.

Our modest observation here is that energy methods are much more powerful in certain hyperbolic spaces because the change in geometry implies that the spectral gap of the Laplacian can be bounded independently of volume. (It is worth pointing out here that there are also other plenty of other apparently-natural hyperbolic settings where the spectral gap decreases with volume, e.g., balls with Neumann boundary conditions of increasing radii \cite[Theorem 5]{eigenvaluestozero}, so Models 1 and 2 considered above have to be chosen carefully.) At any rate, as a consequence of the spectral gap in Models 1 and 2, the proof of the main result is quite short and does not require recent advances on rigorous implementation of Bogoliubov's heuristic. This means that the argument provides no new insight on the Euclidean case.

The result raises some potentially interesting questions for further study in the hyperbolic setting.

(i) Our result on BEC is proved without identifying even the leading order of the energy asymptotics in the dilute limit. The spatial localization that commonly appears in energetic lower bounds is technically more challenging in the hyperbolic world because the local spectral gap of the Laplacian can be smaller than the global gap depending on the choice of boundary conditions. We leave it as an open problem to identify the leading asymptotics of the ground state energy in the hyperbolic setting.

(ii) A more precise analysis of the energy asymptotics would presumably be linked to a hyperbolic rendition of Bogoliubov theory \cite{bogoliubov,bbcs1,bbcs3,bbcs2,brietzke_solovej1, fournais_solovej,fournais_solovej2,hainzl}. This should reveal finer information about excitations above the condensate. It is conceivable that, as in the case of BEC considered here, the price for studying a slightly more complicated geometry is made up by its favorable spectral properties.

(iii) Another natural question concerns the fate of the BEC in the infinite-volume limit at positive temperature. The analogous problem has been resolved in the Euclidean setting, see e.g.\ the recent work \cite{deuchert2020gross} and references therein. This could be of practical relevance in case one finds a significantly larger critical temperature in the hyperbolic setting, keeping in mind that it is now possible to set up quantum gases in hyperbolic structures in the laboratory \cite{hu2019quantum,hyp_lattices1,hyp_lattices2, PRXQuantum.2.017003, saa2021higher,lenggenhager2021electric}. 

(iv) Similarly to Point (iii), one could consider a magnetic field in a hyperbolic geometry analogously to \cite{seiringer,rotating1,quantum_finetti} and others who proved BEC in the presence of magnetic fields in the Euclidean GP setting.


\subsection{Structure of the paper}

This paper is organized as follows.

\begin{itemize}

\item In \Cref{sec:model mainresults}, we state the main abstract results, the upper and lower bounds \cref{th:upper bound coro,th:lower bound}. Afterwards, we apply them to the concrete infinite-volume limits of hyperbolic manifolds described above to derive \cref{th:coro 2d,th:coro 3d,th:coro random}.

\item In \Cref{sec:upper bound}, we prove the upper bound \cref{th:upper bound coro}.  This follows an argument going back to \cite{dyson} in the modern form of \cite{2d_bosegas,greenbook}. The change in geometry leads to a few changes and, similarly to the Euclidean case, we obtain an upper bound on the effective $2$-particle problem by estimating integrals over the fundamental domain by integrals over the universal cover (for us, this is $\IH^d$), cf \eqref{eq:extend integral to Hd}. 

\item In \Cref{sec:lower bound}, we use the spectral gap of the Laplacian and the non-negativity of the potential to prove the lower bound, \cref{th:lower bound}.

\item In \Cref{appendix}
we generalize the scattering length to the hyperbolic setting, defining it via the radius of a hardcore potential, cf. \cref{th:variational}. As in the Euclidean case, we particularly use integration by parts and the inequality \eqref{eq:f_R greater} to estimate $I(f_R)$, $J(f_R)$ and $K(f_R)$. The results are analogous to the Euclidean setting and a key role is played by the harmonic function in the hyperbolic setting \eqref{eq:harmonic solutions}.
\end{itemize}

\section{Models and Main Results}
\label{sec:model mainresults}

\subsection{General facts about hyperbolic Bose gases}

For any $d \geq 2$ let $\IH^d$ denote the $d$-dimensional hyperbolic space. In $d=2$ we will work in the upper-half plane model
\[
\IH^2 = \{ \z_1 + \i \z_2 : \z_1 \in \IR, ~ \z_2 > 0\} \subseteq \IC
\]
equipped with the Riemannian metric
\[
\d s ^2 = \frac{1}{z_2^2}( \d \z_1 ^2 + \d \z_2^2 ).
\]
For dimensions $d \geq 3$ it will be more convenient to work in the hyperboloid model 
\begin{align}
\label{eq:hyperboloid model}
\IH^d := \{ \z \in \IR^{d+1} : \z_0 > 0 \text{ and } q_d(\z) = 1      \}, \quad q_d(\z) := \z_0^2 - \z_1^2 - \ldots - \z_d^2,
\end{align}
equipped with the pullback of the standard Lorentzian metric on $\IR^{d+1}$.
\begin{remark}
Throughout this paper we always use $x$ for elements of the hyperbolic manifolds and $z$ for elements of their universal cover $\IH^d$. 
\end{remark}

Let $X$ be a $d$-dimensional hyperbolic manifold, that is, a complete Riemannian manifold of constant curvature $-1$. Equivalently, $X = \IH^d / \Gamma$, where $\Gamma$ is a discrete subgroup of $\Iso(\IH^d)$ -- the group of isometries of $\IH^d$. We assume that $X$ has finite volume. Denote by $-\Delta \geq 0$ the standard Laplace-Beltrami operator acting on $L^2(X)$. Furthermore, let $V \in L^\infty(\IR_+)$  be a function with compact support and let $R_0 > 0$ such that $\supp V \subseteq [0,R_0]$. For $N \in \IN$ particles consider the Hilbert space of $N$ bosonic particles
\[
\H_N := P^+_N L^2( X^{\times N}),
\]
$P^+_N$ being the symmetrization operator in the $N$ components. In this space we define the Hamiltonian for the Bose gas on $X$ by
\begin{align}
\label{eq:bose gas hamiltonian}
H_N := -\mu \sum_{i=1}^N \Delta_i + \sum_{1 \leq i <j \leq N} V( d(\q_i, \q_j) ),
\end{align}
with domain $\Def(H_N) := P^+_N \Def( \sum_{i=1}^N \Delta_i ) = P^+_N H^2( X^{\times n} )$, where $\Delta_i$ denotes the operator acting as $\Delta$ on the $i$-th component, 
 $d \: X \times X \rightarrow [0,\infty)$ the distance function on $X$, and $d(\q_i, \q_j)$ the multiplication operator by the function
\[
X^n \ni (x_1, \ldots, x_n) \mapsto d(x_i, x_j).
\]
Note that $H_N$ is self-adjoint on $\Def(H_N)$ as $V$ is assumed to be essentially bounded. Furthermore, $H_N$ has a unique normalized ground state $\Psi_0 \in \H_N$ with corresponding ground state energy $E_N$. This can be deduced from the strict positivity of corresponding semigroup (cf. \cite[Theorem XIII.44]{rs4}), which in turn follows from $V \geq 0$ and the  fact that the semigroup associated to the Laplace-Beltrami on connected manifolds is positivity-improving (proven in \cite{positivity_improving_original}, see also \cite[p.139]{positivity_improving_book}).   

The one-particle density matrix $\gamma$ as a bounded operator on $L^2(X)$ is given by the integral kernel
\begin{align}
\label{eq:defn gamma}
\gamma(x,x') :=  \int_{X^{\times(N-1)}} \Psi_0(x,\mathbf{x}) \Psi_0(x', \mathbf{x}) \d \mathbf{x}.
\end{align}
Furthermore, let $\cgs := \vol X ^{-1/2} \ind_X$ be the ground state of $-\Delta$ on $X$, in other words, the normalized constant function on $X$. 

In order to establish BEC in the sense of \eqref{eq:mainresultXL} we use the following abstract lower bound for manifolds $X$ where the Laplacian has a gap. The proof can be found in \Cref{sec:lower bound}.
\begin{prop}[Lower bound]
\label{th:lower bound}
Let $X = \IH^d / \Gamma$ where $\Gamma$ is a discrete subgroup of $\Iso(\IH^d)$ such that $\vol X < \infty$. Assume there exists $\Xi > 0$ such that $-\Delta (\Id- \ket \cgs \bra \cgs) \geq \Xi$. Then we have 
\[
\sc{ \cgs,  \gamma \cgs} \geq 1- \frac{\EN}{N \Xi}.
\]
\end{prop}
Hence, in order to obtain a concrete lower bound, we need now an upper bound for $E_N / N$. This will be now given in terms of a diluteness parameter defined as
\begin{align}
\label{eq:defnY}
\Y := \begin{cases} \rho \ln ((\tanh (a/2))^{-1}) &: d = 2, \\ \rho \tanh a &: d=3,  \end{cases}
\end{align}
where $a$ is the `hyperbolic scattering length' which depends only on the potential $V$ and is defined in \eqref{eq:harmonic solutions}. In fact, we can make  $E_N / N$ arbitrarily small if $\Y$ is small enough. To this end, for any $\varepsilon > 0$ let 
\begin{align}
\label{eq:Y0}
\Y_0(\varepsilon) := \begin{cases}  
\min \left\lbrace 3 \frac{\sqrt{ \frac{2\varepsilon }{3 \mu}  + 1 } -1 }{16 \pi}, (8 \pi  (R_0+1)^2 )^{-1}  \right\rbrace  &: d=2, \\ 
\min \left\lbrace 3 \frac{\sqrt{ \frac{2 \varepsilon }{3 \mu}  + 1 } -1 }{16 \pi e^{2 R_0}}, (8 e^{2R_0} (R_0+1)^2)^{-1}     \right\rbrace     &: d=3.   \end{cases}
\end{align}
Then we obtain the following (see \Cref{sec:upper bound} for the proof).
\begin{thm}[Abstract upper bound]
\label{th:upper bound coro}
Let $X = \IH^d / \Gamma$ where $\Gamma$ is a discrete subgroup of $\Iso(\IH^d)$ such that $\vol X < \infty$. Let $V$ be a potential supported in $[0,R_0]$ with hyperbolic scattering length $a$, and set $E_N = \inf \sigma(H_N)$. 
Given that the diluteness parameter $Y$ (defined as in \eqref{eq:defnY}) satisfies $\Y \leq (8 \pi  (R_0+1)^2 )^{-1}$ in $d=2$ or $\Y \leq (8 e^{2R_0} (R_0+1)^2)^{-1}$ in $d=3$, we have
\begin{align}
\label{eq:upper bounds concrete}
\frac{\EN}{N} \leq \begin{cases} 
16\pi \mu \Y  ( 1 + \frac{8 \pi }{3} \Y) 
  &: d=2, \\ 
 16 \pi \mu e^{2R_0} \Y \left( 1 + \frac{8}{3} \pi e^{2 R_0} \Y  \right) &:  d=3. \end{cases}
\end{align}
In particular, we have $\frac{\EN}{N} \leq \varepsilon$ for all $\Y \leq \Y_0(\varepsilon)$. 

\end{thm}

We now apply the combination of \cref{th:lower bound,th:upper bound coro} to two classes of hyperbolic manifolds which are known to have spectral gaps. In order to be able to obtain a thermodynamic limit, our goal is to find a sequence of manifolds of growing volume tending to infinity with a uniform spectral gap. The first one comprises special non-compact hyperbolic manifolds of finite volume.

\subsection{Modular surfaces}
\label{sec:2d application}

The case of $d=2$ dimensions where  one considers so-called modular surfaces, cf. \cite{sarnak},  is the most well-studied and most thoroughly understood one.
The special linear group $\PSL_2(\IR) := \SL_2(\IR) / \{ \pm \Id \}$ acts on $\IH^2 \subset \IC$ via Möbius transformations
\[
\begin{pmatrix}
a & b \\
c & d
\end{pmatrix} z := \frac{az + b}{cz + d},
\]
and is isomorphic to the group of orientation-preserving isometries of $\IH^2$. The modular surfaces arise by considering the action of a discrete subgroup of $\PSL_2(\IR)$, the modular group $\PSL_2(\IZ) := \SL_2(\IZ) / \{ \pm \Id \}$  and its congruence subgroups. The latter are defined as those subgroups, which contain one of the principal congruence subgroups given by
\[
\Gamma(L) := \{ A \in \SL_2(\IZ) : A = \Id \mod L \}, \quad L \in \IN,
\]
where $\mod L$ is to be understood as taken in each entry of the matrices. 
We can then define for each $L \in \IN$ a hyperbolic surface by
\[
X_L := \IH^2 / \Gamma(L),
\]
A fundamental domain for $X_1$ is given by (cf. \cite[Lemma 2.3.1]{modularforms_book})
\[
F_1 = \{ z \in \IH^2 : \Re z \leq 1/2,~ \abs z \geq 1\},
\]
and from that one can compute directly that $\vol {X_1} = \frac{\pi}{3}$. Furthermore, as
\[
[ \SL_2(\IZ) : \Gamma(L) ] = L^3 \prod_{p \text{ prime},~p|L} \left(1 - \frac{1}{p^2} \right),
\]
see \cite[Exercise 1.2.3(b)]{modularforms_book}, we can infer that 
\[
\vol{X_L} = [ \SL_2(\IZ) : \Gamma(L) ]   \vol{X_1} = L^3 \prod_{p \text{ prime},~p|L} \left(1 - \frac{1}{p^2} \right)\frac{\pi}{3}.
\]
In particular this shows that $\vol{ X_L } \to \infty$, as $L \to \infty$. 

Next, we need a uniform spectral gap. First, one can show that there is a gap of $\frac{1}{4}$ for the continuous spectrum of any hyperbolic surface. 
\begin{prop}[\cite{selberg1,sarnak}]
Let $X$ be a Riemannian surface, that is, $X = \IH^2 / \Gamma$, where $\Gamma$ is any discrete subgroup of the group of isometries of $\IH^2$. Then the continuous spectrum of the Laplacian on $X$ equals $[1/4, \infty)$. 
\end{prop}
It remains to find a similar bound for the lowest non-zero eigenvalue. \name{Selberg} conjectured that one actually has the same lower bound $1/4$ \cite{selberg1}. Although this remains an open problem, there are several proofs for slightly weaker bounds, being sufficient for our application. The to the authors' best knowledge best one is given in the following theorem. 
\begin{thm}[{\cite{kim16functoriality}}]
\label{th:lambda1 2d}
Let $\Gamma$ be a congruence subgroup of $\PSL_2(\IZ)$ and $X = \IH^2 / \Gamma$.
For the smallest non-zero eigenvalue $\lambda_1(X)$ of the Laplacian on $X$ one has
\[
\lambda_1(X) \geq \frac{1}{4} - \left( \frac{7}{64} \right)^2 = \frac{975}{4096}.
\]
\end{thm}
\begin{remark}
\name{Selberg} already proved the bound $\lambda_1(X) \geq \frac{3}{16}$ in \cite{selberg2}. For a discussion of further bounds we refer the reader to \cite{sarnak}. 
\end{remark}
Now, combining \cref{th:lower bound,th:upper bound coro} with \cref{th:lambda1 2d} in the setting of modular surfaces yields the following first application. 
\begin{coro}
\label{th:coro 2d}
Let $R_0 > 0$. For all $\varepsilon > 0$ and all potentials $V$ with $\supp V \subseteq R_0$, scattering length $a$ and all $N,L \in \IN$ satisfying
\[
\rho \ln ( ( \tanh a)^{-1} ) = \frac{N}{\vol{X_L}} \ln ( ( \tanh a)^{-1} ) < \Y_0 \left(  \frac{975}{4096} \varepsilon \right)
\]
where $\Y_0(\cdot)$ is defined in \eqref{eq:Y0}, we have
\[
\sc{ \cgs,  \gamma \cgs} \geq 1 - \varepsilon. 
\]
\end{coro}
\begin{remark}
\label{rem:bec coro}
The statement of \cref{th:coro 2d} implies the double limit statement in \eqref{eq:mainresultXL}. However, it is stronger than \eqref{eq:mainresultXL} in two ways: (a) it is quantitative and (b) the occurrence of BEC only requires an assumption on $\rho$ and $a$, so it also holds for any finite number of particles $N$ and volume $\mathrm{vol}(X_L)$ corresponding to the same density. While we focused on the infinite-volume limit in the introduction for emphasis, the result also applies to finite systems.
\end{remark}

\subsection{Quotients of hyperbolic \texorpdfstring{$3$}{3}-space by congruence subgroups}

The gap of modular surfaces can be generalized to higher dimensions. Here, it is more convenient to work in the hyperboloid model, see \eqref{eq:hyperboloid model}. For a unit ring $R$ let $\SO_{d,1}(R)$  be the group of $R$-valued matrices with determinant one which leave $q_d$ invariant. 
The group of orientation-preserving isometries in $\IH^d$ is then given by $\SO^0_{d,1}(\IR)$, which is defined as the connected component of the identity matrix in $\SO_{d,1}(\IR)$.
\begin{remark}
In \Cref{sec:2d application} we used that $\SO^0_{2,1}(\IR) \cong \PSL_2(\IR)$.
\end{remark}
Now we can consider principal congruence subgroups as follows, see also \cite{arbitrarydimension, burger_sarnak}. Let $\SO^0_{d,1}(\IZ) := \SO^0_{d,1}(\IR) \cap \SO_{d,1}(\IZ)$. Then the principal congruence subgroups can be defined as 
\[
\Gamma_d(L) := \{ A \in \SO^0_{d,1}(\IZ)  :  A =  \Id  \mod L \}, \qquad L \in \IN.
\]
In particular, note that $\Gamma_d(1) =  \SO^0_{d,1}(\IZ)$. A congruence subgroup is then a subgroup of $\SO^0_{d,1}(\IZ)$ which contains $\Gamma_d(L)$ for some $L$. 

Analogously to the 2-dimensional case we then define $X_L := \IH^d  / \Gamma_d(L)$ and obtain hyperbolic manifolds of finite volume. Again, 
\[
\vol{ X_L } =  [ \SO^0_{d,1}(\IZ) : \Gamma_d(L)]   \vol{X_1},
\] 
which equally tends to infinity as $L \to \infty$. 

Finally, we need a variant of \cref{th:lambda1 2d}, i.e., the existence of a gap, for higher dimensions. For $d=3$ this was proven by \name{Sarnak} \cite{sarnak3d}. In \cite{arbitrarydimension} and \cite{li1987poincare} it was first generalized to arbitrary dimension. Other versions for more general algebraic groups can be found in \cite{burger_sarnak} and \cite{clozel2003demonstration}.
\begin{thm}
\label{th:lambda1 3d}
Let $d \geq 3$, $\Gamma$ be a congruence subgroup of $\SO^0_{d,1}(\IZ)$ and $X = \IH^{d} / \Gamma$.
For the smallest non-zero eigenvalue $\lambda_1(X)$ of the Laplacian on $X$ one has
\[
\lambda_1(X) \geq \frac{2d - 3}{4}.
\]
\end{thm}
Then, applying this theorem in combination with \cref{th:upper bound coro} and \cref{th:lower bound} once more for the case $d = 3$ yields the following.  
\begin{coro}
\label{th:coro 3d}
Let $d=3$ and $R_0 > 0$. For all $\varepsilon > 0$ and all potentials $V$ with $\supp V \subseteq [0,R_0]$, scattering length $a$ and all $N,L \in \IN$ satisfying
\[
\rho \ln \tanh a = \frac{N}{\vol{X_L}} \ln \tanh a < \Y_0 \left(  \frac{3}{4} \varepsilon \right)
\]
where $\Y_0(\cdot)$ is defined in \eqref{eq:Y0}, we have
\[
\sc{ \psi_0,  \gamma \psi_0} \geq 1 - \varepsilon. 
\]
\end{coro}



\subsection{Random compact hyperbolic surfaces}
\label{sec:random}

Another possibility is to consider compact hyperbolic manifolds. An analogy of \name{Selberg}'s conjecture in this case is only known in a probabilistic sense and leads to the theory of compact random hyperbolic surfaces as developed by \name{Mirzakhani}.
Recent surveys for this topic can be found in \cite{wright_survey,monk_thesis}.

A compact hyperbolic surface is given by $\IH^2 / \Gamma$, where  $\Gamma \subset \operatorname{PSL}_2(\IR)$ is a discrete and co-compact subgroup. For $g \in \IN$ let
\[
\Mg := \text{ compact hyperbolic surfaces of genus } g ~/ \text{ isometries },
\]
the so-called \textit{moduli space}, which can be also represented as a quotient of the Teichmüller space by some group action \cite[Section 2]{mirzakhani}. On $\mathcal{M}_g$ one can construct a probability measure $\Pwg$ originating from a natural symplectic form on $\Mg$, the so-called Weil-Petersson form \cite[2.8]{wright_survey}. By the Gauss-Bonnet theorem the volume of any $X \in \Mg$ equals $2\pi(2g-2)$ and therefore we can consider $g\to\infty$ for an infinite volume limit. 

In this limit an analog of \name{Selberg}'s $1/4$ conjecture was formulated in \cite{wright_survey}:
\begin{align}
\label{eq:selberg conjecture random} 
\Pwg\left( X \in \Mg : \lambda_1(X) \geq \frac{1}{4} - \alpha\right)  \underset{g\to\infty}{\stackrel?\rightarrow} 1 \qquad \text{ for all } \alpha > 0.
\end{align}
As in the deterministic case, this remains an open problem but several weaker results have been established as well. The currently best one is the following. 
\begin{thm}[{\cite{lambda_positive_random1,lambda_positive_random2}}]
\label{th:lambda positive random}
We have for all $\alpha > 0$
\[
\lim_{g\to\infty} \Pwg\left( X \in \Mg :\lambda_1(X) \geq \frac{3}{16} - \alpha\right) = 1.
\]
\end{thm}
\begin{remark}
This improves a famous result by \name{Mirzakhani} {\cite[Theorem 4.8]{mirzakhani}}, where she showed the same with constant $\frac{1}{4} \left(  \frac{\ln 2}{2\pi + \ln 2} \right)^2 \approx 0.02$ instead of $\frac{3}{16}$.
\end{remark}
\begin{remark}
In other settings of random hyperbolic manifolds, namely for conformally compact infinite area hyperbolic surfaces \cite{magee2021extension} and for  finite area non-compact hyperbolic surfaces \cite{hide2021near} the lower bound $\frac{3}{16}$ in \cref{th:lambda positive random} could actually be  improved to $\frac{1}{4}$, see also references therein. 
\end{remark}

\begin{coro}
\label{th:coro random}
Let $\alpha > 0$ and  $\xi < 1$. Then there exists $g_0 \in \IN$ such that all $g \geq g_0$  there is a measurable set $\mathcal A_g$ with
\[
 \Pwg( \mathcal A_g ) \geq \xi
\]
such that for all $X \in \mathcal A_g$, $R_0 > 0$, $\varepsilon >0$, and all potentials $V$ with $\supp V \subseteq R_0$, scattering length $a$ and and $N \in \IN$ satisfying
\[
\rho \ln ( ( \tanh a)^{-1} ) = \frac{N}{ 2\pi(2g-2) } \ln ( ( \tanh a)^{-1} ) < \Y_0 \left( \left(\frac{3}{16} -\alpha \right) \varepsilon \right)
\]
where $\Y_0(\cdot)$ is defined in \eqref{eq:Y0}, we have
\[
\sc{ \cgs,  \gamma \cgs}  \geq 1 - \varepsilon.
\]
\end{coro}
\begin{proof}
For given $\alpha > 0$ let  
\[
\mathcal A_g := \left\{ X \in \Mg :\lambda_1(X) \geq \frac{3}{16} - \alpha\right \}.
\]
Then we use \cref{th:lambda positive random} and find for $\xi < 1$ a $g_0$ such that $\Pwg( \mathcal A_g ) \geq \xi$ for all $g \geq g_0$.

Now, for $X \in \mathcal A_g$ and under the given assumptions we have
\begin{align}
\label{eq:first step}
\lambda_1(X)  \geq \frac{3}{16} - \alpha  \geq \frac{\EN}{\varepsilon N} 
\end{align}
by \cref{th:upper bound coro}. Thus, by \cref{th:lower bound}
\[
\sc{ \cgs, \gamma  \cgs} \geq 1 - \frac{\EN}{N \lambda_1(X)} \overset{\eqref{eq:first step}}\geq 1-  \varepsilon . \qedhere
\]
\end{proof}
\begin{remark}
The two properties (a) and (b) described in Remark \ref{rem:bec coro} also apply to \cref{th:coro random}.
\end{remark}

%
%

\section{Upper Bound}
\label{sec:upper bound}
In this part we give the proof of \cref{th:upper bound coro}.
First, we show an abstract form of an upper bound, which is in complete analogy with \cite[Section 2.1]{greenbook}, cf. also \cite[(2.7)]{2d_bosegas}. For a function $f$ on $[0,\infty)$ we define a trial function $\Psi \in L^2(\H_N)$ by
\begin{align}
\label{eq:variational ansatz}
\Psi(x_1, \ldots, x_N) := \prod_{i=2}^N F_i(x_1, \ldots, x_i),
\end{align}
where
\begin{align*}
F_i(x_1, \ldots, x_i) &:= f( t_i(x_1, \ldots, x_{i-1}) ), \\
t_i(x_1, \ldots, x_i) &:= \min \{ d(x_i, x_j) : j = 1, \ldots, i-1\}.
\end{align*}
\begin{prop}
\label{th:first estimate}
For any non-decreasing function $f$ on $[0,\infty)$ let $\Psi$ given by \eqref{eq:variational ansatz}. Let $\rho := N /  \vol X$. Then we have
\[
\frac{\sc{ \Psi, H_N \Psi}}{\nn{\Psi}^2} \leq  \frac{N}{(1 - \rho I(f))^2} \left( \rho J(f) + \frac{2}{3} \mu (\rho K(f))^2 \right),
\]
given that the integrals
\begin{align*}
I(f) &:= \int_{\IH^d} ( 1 - f(\absh{\z})^2  )  \d \z  \\
J(f) &:= \int_{\IH^d} \left( \mu f'(\absh \z)^2 + \frac{1}{2} V(\absh \z) \abs{f(\absh \z)}^2 \right) \d \z, \\
K(f) &:=\int_{\IH^d} f(\absh \z) f'(\absh \z) \d \z,
\end{align*}
for any $o \in \IH^d$ chosen as an origin,
are finite and $\rho I(f) < 1$. 
\end{prop}
\begin{proof}
The proof is analogous to the Euclidean case \cite{greenbook} with some modifications for the hyperbolic setting. For a function $\Phi \: X^{\times N} \rightarrow \IC$ let $\nabla_k$ denote the gradient on the manifold $X$ with respect to the $k$-th component, that is, for each $\mathbf x  = (x_1, \ldots, x_N) \in X^{\times N}$, we get an element $\nabla_k \Phi( \mathbf x) \in T_{x_k} X$ if $\Phi$ is smooth enough around $\mathbf x$. We write $\sc{ \cdot, \cdot}_{T_x X} \: T_x X \times T_x X \rightarrow \IR$ for the Riemannian metric of $X$ at the point $x$, and $\nn{ \cdot }_{T_x X}$ for the corresponding norm on $T_x X$.  For notational convenience we will mostly drop the argument $\mathbf x$.   

By the chain rule we get for almost all $\mathbf x \in  X^{\times N}$
\begin{align*}
\nabla_k \Psi &=  \sum_{i \geq k} \frac{\Psi}{F_i} f'(t_i ) \nabla_k t_i = \sum_{i \geq k}  \frac{\Psi}{F_i}  f'(t_i )  \nabla_k d(x_i, x_{i^*} ),
\end{align*}
where
$x_{i^*}$ denotes the nearest neighbor among the points $x_1, \ldots, x_{i-1}$. Therefore,
\begin{align*}
\sum_{k=1}^N \nn{ \nabla_k \Psi }_{T_{x_k}X}^2 
&=  \sum_{i=1}^N  \frac{\abs \Psi^2}{F_i^2} f'(t_i)^2 \sum_{k =1}^i \nn{ \nabla_k d(x_i, x_{i^*} ) }_{T_{x_k}X}^2  \\ &\qquad + 2 \sum_{k=1}^N \sum_{j > i \geq k}   \frac{\abs \Psi ^2 }{F_i F_j } f'(t_i) f'(t_j) \sc{ \nabla_k d(x_i, x_{i^*} )  , \nabla_k d(x_j, x_{j^*} ) }_{T_{x_k}X},
\end{align*}
where we use that there is a unique nearest neighbor almost everywhere and that we have to sum over ordered pairs. Since $\nn{ \nabla_x d(x,y) }_{T_x X} \leq 1$ almost everywhere, observe that $\nn{ \nabla_k d(x_i,x_{i^*}) }_{T_{x_k} X} \leq \epsilon_{ik}$ and $\sum_{k} \epsilon_{ik} \leq 2$ for almost every $\mathbf x$, where
\[
\epsilon_{ik} := \begin{cases}  1 &: i=k \text{ or } t_i = d(x_i, x_k), \\
  0  &: \text{else.}   \end{cases}
\]
 Thus, we arrive at
\begin{align}
\frac{\sc{ \Psi, H_N \Psi}}{\nn{\Psi}^2} &\leq 2 \mu \sum_{i=1}^N\frac{ \int \abs \Psi^2 F^{-2}_i f'(t_i)^2 }{\int \abs \Psi^2 } + \sum_{i<j} \frac{\int \abs \Psi^2 V(d(x_i,x_j)) }{\int \abs \Psi^2 } \label{eq:diagonal term} \\ &\qquad  + 2\mu \sum_{k=1}^N \sum_{j > i \geq k}  \frac{ \int \abs{ \epsilon_{ik} \epsilon_{jk} }  \frac{\abs \Psi ^2 }{F_i F_j } f'(t_i) f'(t_j)  }{\int \abs{ \Psi}^2 }.  \label{eq:non diagonal term}
\end{align}
Now, for $j < i$ we define $F_{i,j}$ in the same way as $F_i$ with the only difference that we omit the point $x_j$ in the consideration of the nearest neighbors. Likewise, we define $F_{i,jk}$ by omitting $x_j$ and $x_k$. Then $F_{i,j}$ does not depend on $x_j$ and $F_{i,jk}$ does not depend on $x_j$ and $x_k$. Furthermore, since $f$ is monotonically increasing and $0 \leq f \leq 1$, we have
\begin{align}
F_{j+1}^2 \cdots F_{i-1}^2 F_{i+1}^2 \cdots F_N^2 &\leq F_{j+1,j}^2 \cdots F_{i-1,j}^2 F_{i+1,ij}^2 \cdots F_{N,ij}^2,  \label{eq:Fj upper bound} \\
\begin{split}
 F_{j}^2 \cdots F_N^2 &\geq \left(  1 - \sum_{k=1,\not=i,j}^N (1- f(d(x_j, x_k))^2)   \right) F_{j+1,j}^2 \cdots F_{i-1,j}^2  \label{eq:Fj lower bound} \\ \times &\left(  1 - \sum_{k=1,\not=i}^N (1- f(d(x_i, x_k))^2)   \right)   F_{i+1,ij}^2 \cdots F_{N,ij}^2 .
\end{split}
\end{align}
Furthermore, we trivially find
\begin{align}
 f'(t_i)^2 \nn{ \eta_i }^2 &\leq \sum_{j=1}^{i-1} f'(d( x_i, x_j))^2 [\nabla_i d(x_i, x_j)]^2 \label{eq:Fi fd prime bound}   , \\
  F_i &\leq f(d(x_i,x_j)) \label{eq:Fi fd bound} . 
\end{align}
Now, the numerator of the right-hand side in \eqref{eq:diagonal term} can be estimated from above using \eqref{eq:Fj upper bound} together with \eqref{eq:Fi fd prime bound} and \eqref{eq:Fi fd bound},
\begin{align}
&\sum_{i=1}^N  \int \abs \Psi^2 F^{-2}_i f'(t_i)^2 + \sum_{j<i} \int \abs \Psi^2 V(d(x_i,x_j)) \\
& \leq 2\mu \sum_{j < i}  \int F_1^2 \ldots F_{j-1}^2  F_{j+1,j}^2 \cdots F_{i-1,j}^2 F_{i+1,ij}^2 \cdots F_{N,ij}^2 \d x_{1,\ldots, N, ij} \\ 
&\quad\times \int \left(  2 \mu f'(d(x_i ,x_j))^2  +  f(d(x_i ,x_j))^2 V(d(x_i,x_j)) \right) \d x_i \d x_j, \label{eq:term kinetic}
\end{align}
where $\d x_{1,\ldots, N, ij}$ denotes the integration over all $x_1, \ldots ,x_N$ except $x_i$ and $x_j$. 
For the denominator we use \eqref{eq:Fj lower bound} and obtain similarly
\begin{align*}
\nn{\Psi}^2 &\geq  \int F_1^2 \cdots F_{j-1}^2  F_{j+1,j}^2 \cdots F_{i-1,j}^2  F_{i+1,ij}^2 \cdots F_{N,ij}^2   \\ 
&\qquad \times  \left(  \vol X - \sum_{k=1,\not=i,j}^N \int_X  (1- f(d(x_j, x_k))^2)  \d x_j  \right)  \\
&\qquad \times \left(  \vol X - \sum_{k=1,\not=i}^N  \int_X  (1- f(d(x_i, x_k))^2)  \d x_i   \right)  \d x_{1,\ldots, N, ij} .
\end{align*}

Now, we use that 
\begin{align}
\label{eq:extend integral to Hd}
\int_{X = \IH^d / \Gamma} g(d(x, x_0)) \d x \leq \int_{\IH^d} g(\absh{\z}) \d \z
\end{align}
for any positive function $g$ defined on $\IH^d$ and all $x_0 \in X$, $o \in \IH^d$. This yields
\begin{align*}
\eqref{eq:term kinetic} &\leq \int_X  \int_{\IH^d} \left( 2 \mu f'(\absh{x_j})^2  +  f(\absh{x_j})^2 V(\absh{x_j}) \right)  \d \z \d x_j  = 2 \vol X J(f),
\end{align*}
and for all $k$,
\begin{align*}
\int_X  (1- f(d(x_j, x_k))^2)  \d x_j  \leq \int_{\IH^d}  (1- f(\absh \z)^2)  \d \z = I(f).
\end{align*}
The integral over $\d x_{1,\ldots, N, ij}$ cancels in the numerator and denominator and we obtain
\begin{align*}
 2 \mu &\sum_{i=1}^N \frac{ \int \abs \Psi^2 F^{-2}_i f'(t_i)^2 }{\int \abs \Psi^2 } + \sum_{i<j} \frac{\int \abs \Psi^2 V(d(x_i,x_j)) }{\int \abs \Psi^2 } 
\\ &\leq  \frac{N(N-1)}{2} \frac{2 \vol X J(f)}{(\vol X - (N-1)I(f))^2 }.
\end{align*}

Next, we estimate the non-diagonal term \eqref{eq:non diagonal term}, cf. \cite{bosons_in_trap}. We get the same cancellations in the numerator and denominator up to the term
\begin{align*}
2\mu &\sum_{k=1}^N \sum_{j > i \geq k}  \int_{X \times X} \abs{ \epsilon_{ik} \epsilon_{jk} }  f(t_i) f(t_j) f'(t_i) f'(t_j) \d x_i \d x_j \\  
& \leq 4\mu \sum_{k=1}^N \sum_{j > i > k}  \int_{X \times X}   f(d(x_i, x_k)) f(d(x_j, x_k))  f'(d(x_i, x_k)) f'(d(x_j, x_k)) \d x_i \d x_j   \\ 
 & =  \frac{2}{3}\mu N (N-1) (N-2) K(f)^2.
\end{align*}
This shows the desired bounds.
\end{proof}
%
%

%
%
%

Our choice for $f$ in definition of the $F_i$ \eqref{eq:variational ansatz} will be $f_R$, $R > 0$, given as in \eqref{eq:harmonic solutions}. Then we have $J(f_R) = E_R$, which is explicitly computed in \cref{th:variational}. Therefore, it remains to find explicit bounds for $I(f_R)$ and $K(f_R)$, which is the content of the following two lemmas. 
\begin{lemma}
\label{th:I estimate}
For all $R > R_0$,
\[
I(f_R) \leq \begin{cases} \frac{2\pi }{  \ln \frac{\tanh(R/2)}{\tanh(a/2)}  } (R^2 - a^2) &: d=2, \\    \frac{ 4 \pi \tanh a }{\tanh R - \tanh a } \tanh R (R^2 - a^2) &: d=3. \end{cases} 
\]
\end{lemma}
\begin{proof}
Using hyperbolic polar coordinates, we get 
\begin{align*}
I(f_R) &= \vol{\IS^{d-1}}  \int_0^R (1 - f_R(r)^2)  \sinh^{d-1} r \d r \\
&\leq \vol{\IS^{d-1}} \int_0^a \sinh^{d-1} r \d r +  \vol{\IS^{d-1}} \int_a^R \left( 1 - \frac{f_\infty(r)^2}{f_\infty(R)^2} \right) \sinh^{d-1} r \d r \\
&= \vol{\IS^{d-1}}\int_0^R \sinh^{d-1} r \d r - \frac{\vol{\IS^{d-1}}}{f_\infty(R)^2} \int_a^R  f_\infty(r)^2  \sinh^{d-1} r \d r .
\end{align*}
Let $u(r) := \int_0^r \sinh^{d-1}(r') \d r'$. 
With integration by parts the second term can be expressed as
\begin{align*}
 \int_a^R  f_\infty(r)^2  \sinh^{d-1}r \d r  &= [ f_\infty(r)^2 u(r)]_a^R - \int_a^R 2 f_\infty(r) f_\infty'(r) u(r)  d r \\
&= f_\infty(R)^2 \int_0^R \sinh^{d-1}r \d r - \int_a^R 2 f_\infty(r) f_\infty'(r) u(r)  d r .
\end{align*}
Thus, using that $u(r) \leq r \sinh^{d-1}r$ and $f_\infty'(r) \sinh^{d-1}r = \Cda$ is independent of $r$ (cf. \Cref{rem:cda}), 
\begin{align*}
I(f_R) &\leq \frac{ 2 \vol{\IS^{d-1}} }{f_\infty(R)^2} \int_a^R f_\infty(r) f_\infty'(r) u(r)  \d r \\
&\leq \frac{  2 \Cda \vol{\IS^{d-1}}}{f_\infty(R)^2} \int_a^R \underbrace{ f_\infty(r) }_{\leq f_\infty(R)} r   \d r \\
&\leq \frac{  \Cda \vol{\IS^{d-1}} }{f_\infty(R)} (R^2 - a^2).
\end{align*}
Note that we have $C_2(a) = 1$ and $C_3(a) = \tanh a$. 
\end{proof}

\begin{lemma}
\label{th:K estimate}
For all $R > R_0$,
\[
K(f_R) \leq \begin{cases} \frac{2 \pi R}{  \ln \frac{\tanh(R/2)}{\tanh(a/2)}   } &: d=2, \\   \frac{4 \pi \tanh a R}{ 1 - \frac{\tanh a}{\tanh R}} &: d=3. \end{cases} 
\]
\end{lemma}
\begin{proof}
Using $f_R'(r) f_R(r) = \frac{1}{2} (f_R(r)^2)'$, partial integration and \eqref{eq:f_R greater}, we obtain 
\begin{align*}
K(f_R) 
&= \frac{ \vol{\IS^{d-1}} }{2}  \int_0^R (f_R(r)^2)' \sinh^{d-1} r \d r    \\
&=     \frac{ \vol{\IS^{d-1}} }{2} f_R(R)^2 \sinh^{d-1}(R) -     \frac{ \vol{\IS^{d-1}} }{2} \int_0^R f_R(r)^2  (\sinh^{d-1} r)' \d r    \\
&\leq     \frac{ \vol{\IS^{d-1}} }{2} f_R(R)^2 \sinh^{d-1}(R) -     \frac{ \vol{\IS^{d-1}} }{2 f_\infty(R)^2} \int_a^R f_\infty(r)^2  (\sinh^{d-1} r)' \d r    \\
&= \frac{ \vol{\IS^{d-1}} }{f_\infty(R)^2} \int_a^R f_\infty(r) f_\infty'(r) \sinh^{d-1} r \d r .
\end{align*}
Using again that $f_\infty'(r) \sinh^{d-1}r = \Cda$ is independent of $r$, we conclude
\begin{align*}
K(f_R) \leq \frac{ \Cda \vol{\IS^{d-1}} }{f_\infty(R)^2} \int_a^R f_\infty(r) \d r \leq \frac{ \Cda \vol{\IS^{d-1}} }{f_\infty(R)}.
\end{align*}
Plugging in $\vol{\IS^1} = 2\pi$, $\vol{\IS^2} = 4\pi$, $C_2(a) = 1$, $C_3(a) = \tanh a$, and $f_\infty(R) =  \ln \left( \frac{\tanh(R/2)}{\tanh(a/2)} \right)$ for  $d = 2$, $f_\infty(R) =  1 - \frac{\tanh a}{\tanh R}$ for $d=3$ yields the claimed estimates. 
\end{proof}

\begin{remark}
\label{rem:abstract form}
One can write the estimates from \cref{th:I estimate,th:K estimate} in a dimension-independent way as
\begin{align*}
I(f_R) &\leq  f_\infty'(R) \sinh^{d-1} R   \frac{  \vol{\IS^{d-1}} }{f_\infty(R)} (R^2 - a^2), \\
K(f_R) &\leq   f_\infty'(R) \sinh^{d-1} R  \frac{\vol{\IS^{d-1}} }{f_\infty(R)},
\end{align*}
as it can be seen in the proofs. 
\end{remark}

\begin{prop}
\label{th:upper bound}
Let $R \geq R_0$ and $R > a$. In $d=2$ we have for all $\rho$ and $a$ 
\[
\frac{\EN}{N} \leq \frac{2\pi \rho \mu}{\left(1 -  \frac{ 2\pi \rho  }{\ln \frac{\tanh(R/2)}{\tanh(a/2)}} (R^2 - a^2) \right)^2 \ln \frac{\tanh(R/2)}{\tanh(a/2)} } \left( 1 + \frac{4}{3} \frac{\pi \rho}{\ln \frac{\tanh(R/2)}{\tanh(a/2)}} \right),
\]
provided that $\frac{ 2\pi \rho  }{\ln \frac{\tanh(R/2)}{\tanh(a/2)}} (R^2 - a^2) < 1$ and in $d=3$ we have for all $\rho$ and $a$ 
\[
\frac{\EN}{N} \leq \frac{ 4 \pi \rho \mu \tanh a \tanh R }{\left(1 -  4\pi \rho \tanh a \frac{  (R^2 - a^2) \tanh R }{\tanh R - \tanh a}  \right)^2 (\tanh R - \tanh a)} \left( 1 + \frac{8}{3}  \frac{\pi \rho \tanh a \tanh R}{\tanh R - \tanh a} \right),
\]
provided that $4\pi \rho \tanh a \frac{  (R^2 - a^2) \tanh R }{\tanh R - \tanh a} < 1$.
%
%
\end{prop}
\begin{proof}
Plugging in the concrete upper bounds of \cref{th:I estimate}, \cref{th:K estimate} (in the form of \Cref{rem:abstract form}) and \cref{th:variational} in \cref{th:first estimate} yields
\begin{align*}
\frac{\sc{ \Psi, H_N \Psi}}{N \nn{\Psi}^2} &\leq \frac{ \rho \Cda \mu \vol{\IS^{d-1}}}{\left(1 -  \rho  \Cda \frac{  \vol{\IS^{d-1}} }{f_\infty(R)} (R^2 - a^2) \right)^2 f_\infty(R) } \left( 1 + \frac{2}{3} \rho \frac{\Cda \vol{\IS^{d-1}}}{f_\infty(R)} \right).
\end{align*}
By using the values for $\Cda$ and $f_\infty(R)$ for $d=2$ and $d=3$  one obtains the claimed upper bounds. 
\end{proof}

\begin{proof}[Proof of \cref{th:upper bound coro}]
Choose $R := \max\{ R_0, a+ 1\}$, which is eligible in \cref{th:upper bound}. Then we find, using $a \leq R_0$ that $R^2-a^2 \leq (R_0+1)^2$. Furthermore, for $d=2$ we have
\begin{align*}
\frac{\rho}{\ln \frac{\tanh (R/2)}{\tanh(a/2)}} &= \frac{\rho}{\ln (\tanh(a/2)^{-1}) \left( 1 - \frac{\ln \tanh(R/2)}{\ln \tanh( a/2) } \right)} \\ 
&\leq \frac{1}{1 - \frac{  \ln \tanh((a+1)/2)}{ \ln \tanh( a/2) } }  \Y \\
&\leq \frac{1}{1 - e^{-1}} \Y \leq 2 \Y, 
\end{align*}
and for $d=3$,
\begin{align*}
\frac{\tanh R}{\tanh R - \tanh a} \leq \frac{\tanh(a+1)}{\tanh(a+1) - \tanh a}  \leq e^{2a} \leq e^{2R_0}.
\end{align*}
Using these estimates, the upper bounds of \cref{th:upper bound}  simplify as follows: 
\[
\frac{\EN}{N} \leq \begin{cases} 
\frac{ 4\pi \mu \Y }{ (1 - 4 \pi  (R_0+1)^2  \Y )^2  } ( 1 + \frac{8 \pi }{3} \Y) 
  &: d=2, \\ 
\frac{ 4 \pi \mu e^{2R_0} \Y }{\left(1 -  4 \pi e^{2R_0} (R_0+1)^2  \Y   \right)^2} \left( 1 + \frac{8\pi }{3} e^{2 R_0} \Y  \right) &:  d=3. \end{cases}
\]
If we assume $4 \pi  (R_0+1)^2  \Y \leq \frac{1}{2}$ and $4 \pi e^{2R_0} (R_0+1)^2   \Y  \leq \frac{1}{2}$, respectively, we get \eqref{eq:upper bounds concrete}. 

For the choice of $\Y_0(\varepsilon)$ note that the inequality $aY(1 + bY) \leq c$ has the solution
\[
\Y \leq \frac{\sqrt{ 4b c/a +1 } -1 }{2b}
\]
for $\Y \geq 0$.
\end{proof}


\section{Lower Bound}
\label{sec:lower bound}

In this section we prove the lower bound (\cref{th:lower bound}). 
Recall that 
 $\Psi_0 \in L^2(X^{\times N})$ is the ground state of the operator $H_N$ \eqref{eq:bose gas hamiltonian} and $\cgs$ the ground state of $-\Delta$ on $X$. Furthermore, the one-particle density matrix $\gamma$ was defined in \eqref{eq:defn gamma}. 

\begin{proof}[Proof of \cref{th:lower bound}]
Since $V\geq 0$, we have
\[
\tr ( - \Delta \gamma )= \sc{ \Psi_0, - \frac{1}{N} \sum_{i=1}^N \Delta_i \Psi_0} \leq \frac{E_N}{N}.
\]
Let $P_m := \ind_{[0,m]}(-\Delta)$ be the spectral projection of  $-\Delta$ to all values smaller than $m$, which makes $-\Delta P_m$ bounded. 
By dominated convergence using that $\tr ( -\Delta P_m \gamma ) \leq E_N /N$ we see that $\tr( (-\Delta + \Delta P_m) \gamma) \to 0$, $m \to \infty$. Now, by the spectral theorem we can write
\[
\gamma = \sum_n p_n \sc{ \cdot, \phi_n} \phi_n, \qquad \sum_n p_n =1, 
\]
with $(\phi_n)$ being an orthonormal basis  of $L^2(X)$. 
We obtain 
\begin{align*}
\tr(-\Delta \gamma) &=  \lim_{m \to \infty} \sum_n \sc{ \phi_n, -\Delta P_m \gamma \phi_n} \\
&= \lim_{m \to \infty} \sum_n p_n \sc{ \phi_n, -\Delta P_m \phi_n} \\
&\geq \Xi \sum_n p_n \nn{  \ket \cgs \bra \cgs ^\perp \phi_n }^2.
\end{align*}
Thus,
\begin{align*}
\sc{ \cgs, \gamma  \cgs} &= \sum_n p_n \abs { \sc{ \cgs, \phi_n} }^2   = \sum_n p_n \left(1 - \nn{\ket \cgs \bra \cgs^\perp \phi_n}^2 \right) \\ &= 1 -  \sum_n p_n  \nn{\ket \cgs \bra \cgs^\perp \phi_n}^2 \geq 1- \frac{E_N}{N \Xi}. \qedhere
\end{align*}
\end{proof}

\hyphenation{Deu-chert}
\hyphenation{Grund-leh-ren}
\section*{Acknowledgments}
The authors thank Christian Brennecke, Matthew de Courcy-Ireland, Andreas Deuchert, S{\o}ren Fournais, and Christian Hainzl for useful comments.
\appendix

\section{Variational principle}


In this part we show existence and uniqueness of the ground state for the key two-particle scattering problem in the hyperbolic setting. This will be used in the choice of the $N$-particle test functions in the upper bound in \Cref{sec:upper bound}. We also define a `hyperbolic scattering length' $a$. As in the Euclidean case it will correspond to the radius of a hardcore potential with the same scattering behavior. The arguments follow closely those in \cite{2d_bosegas,greenbook}.  

Let $V \: [0,\infty) \rightarrow [0,\infty)$ be a measurable function with essential compact support and let $R_0 > 0$ such that $\supp V \subseteq [0,R_0]$. Let $o \in \IH^d$ be fixed. 
For $R > R_0$ and $\phi \in H^1(B_R(o))$, we define the functional
\begin{align*}
\Ee_R(\phi) &:= \int_{B_R(o) \subseteq \IH^d} \left( \abs{ \nabla \phi(\z) }^2 + \frac{1}{2} V(d(o,\z)) \abs{ \phi(\z)}^2  \right) \d \z.
\end{align*}
\begin{remark}
With this functional we can describe two-particle energies on a $d$-dimensional hyperbolic manifold $X$, cf. the proof of \cref{th:first estimate}. 
\end{remark}

\begin{thm}
\label{th:variational}
In the class of functions $\phi \in H^1(B_R(o))$ with $\phi(\z) = 1$ for a.e. $\z$ with $\absh \z = 1$ there  exists a unique minimizer $\phi_{R}$ of $\Ee_R$. It is spherically symmetric, non-zero and satisfies the Euler-Lagrange equation
\begin{align}
\label{eq:euler lagrange}
-\mu \Delta \phi(\z) + \frac{1}{2} V(\z)\phi(\z)  = 0.
\end{align}
For $R_0 < r < R$, we have $\phi(\z) = f_R(\absh \z)$ with 
\begin{align}
\label{eq:harmonic solutions}
f_{R}(r) :=   \frac{f_\infty(r)}{f_\infty(R)}, \qquad f_\infty(r) := \begin{cases}  \ln \left( \frac{\tanh(r/2)}{\tanh(a/2)} \right) &: d = 2, \\ 1 - \frac{\tanh a}{\tanh r} &: d = 3, \end{cases}
\end{align}
for some number $a > 0$.  The energy corresponding to $\phi$ is given by 
\[
E_R := \Ee_R(\phi) =  \frac{ \mu f'(R) \sinh^{d-1} R  \vol{\IS^{d-1}}}{ f_\infty(R) }  =  \begin{cases} \frac{2 \pi \mu }{   \ln \left( \frac{\tanh(R/2)}{\tanh(a/2)} \right) } &: d = 2, \\ \frac{ 4\pi\mu a }{ 1 - \frac{\tanh a}{\tanh R} }  &: d = 3. \end{cases}
\]
Finally, we have that $f_R$ is non-decreasing and
\begin{align}
\label{eq:f_R greater}
f_R(r) \geq \frac{f_\infty(r)}{f_\infty(R)} \text{ for all } r \geq a. 
\end{align}

\end{thm}
\begin{remark}
\label{rem:cda}
\begin{enumerate}[label=(\alph*)]
\item Notice that we indeed defined $a$ in such a way that $f_\infty(a) = 0$, i.e., if $V$ is a hardcore potential with radius $R_0$, then $a = R_0$.
\item As $f$ is an indefinite integral of $( \sinh^{d-1} )^{-1}$, the quantity  $f'(R) \sinh^{d-1} R$ only depends on $a$ (or $V$) but not on  $R$. Therefore, we also write 
\begin{align*}
\Cda := f'(R) \sinh^{d-1} R.
\end{align*}
\end{enumerate}
\end{remark}
\newcommand{\wphi}{\widehat{\phi}}
\begin{proof}
First, we show that we can restrict to non-negative and spherically symmetric functions as minimizers. Let  $f,g$ real-valued functions on $\IR$. Then we find (cf. \cite[Theorem 7.8]{ll_analysis})
\[
 \left( \frac{d}{dr} \sqrt{ f^2 + g^2 }  \right)^2 +  \frac{ ( fg' - g f')^2 }{f^2 + g^2} = (f')^2 + (g')^2
\]
for all points where they are differentiable and $f^2 + g^2 > 0$. Thus, we have a.e.
\[
 \abs{ \frac{d}{dr} \sqrt{ f^2 + g^2 }   }^2  \leq  \left(  \abs{  f' }^2  + \abs{ g' }^2 \right).
\]
Furthermore, 
\[
\nabla \sqrt{ f(d(o,\z)) }  = \frac{d}{dr} \sqrt{ f(r) }\bigg|_{r = d(0,\z)} \nabla d(o,\z)
\]
for a.e. $z \in \IH^d$.
Hence, the map $f \mapsto \int_{\IH^d} \left( \nabla \sqrt{ f(d(o,\z))} \right)^2 \d \z$, is convex. For $\phi \in H^1(B_R(o))$ let $\widetilde \phi$ be the spherically symmetric function given by the square root of the spherical average of $\phi^2$. By the generalized Jensen inequality for probability measures, we obtain $\nn{ \nabla \widetilde \phi}^2 \leq \nn{\nabla \phi}^2$ and thus also $\Ee_R(\widetilde \phi) \leq \Ee_R(\phi)$ because the potential is assumed to be spherically symmetric as well. 

\textit{Existence of a minimizer:}
As $\Ee_R$ is bounded from below, there exists a minimizing sequence of spherically symmetric $(\phi_n)$ in $H^1(B_R(o))$ with $\phi_n(\z) = 1$ for a.e. $\z$ with $\absh \z = 1$ and all $n$. Define $\widehat{\phi}_n \in H^1(\IH^d)$ by
$\widehat{\phi}_n(\z) := \phi_n(\z)$ for $\z \in B_R(o)$ and $\phi_n(\z) = h(d(o,\z))$ for some $h \in C^\infty(\IR_+)$ with $h(r) = 1$ for $r < R+1$ and $h(r) = 0$ for $r > 2R+1$. As $\sup_{n} \nn{\wphi_n }_{H^1(\IH^d)} < \infty$ (and because $H^1(\IH^d)$ is reflexive, cf. \cite[Proposition 2.4]{sobolev_manifolds}), one can find a subsequence $(\wphi_{n_k})$ in $H^1(\IH^d)$ which converges weakly in $H^1(\IH^d)$ to some $\wphi \in H^1(\IH^d)$, which is rotationally symmetric. We then have that $(\phi_{n_k})$ also converges weakly to $\phi := \wphi|_{B_R(o)} \in H^1(B_R(o))$. One gets $\phi(\z) = 1$ for a.e. $\z$ with $\absh \z = 1$ because the radial part is continuous outside of the origin, and $\wphi(\z) = 1$ for $d(o,\z) \in (R,R+1)$. By equivalence of lower semicontinuity and weak lower semicontinuity for convex functions, we obtain $\lim_{k \to \infty} \Ee_R(\phi_{n_k}) \geq \Ee_R(\phi)$ and therefore, $\phi$ is a minimizer. 


The Euler-Lagrange equation \eqref{eq:euler lagrange} follows by considering $\frac{d}{\d \delta }|_{\delta = 0} \Ee_R(\phi + \delta \psi) = 0$ for all infinitely differentiable functions $\psi$ which vanish for all $\z$ with $d(o,\z) \geq R$. Furthermore, \eqref{eq:euler lagrange} can be written down for the radial part $f_R$ on $(0,R)$ given by $f_R(d(o,\z)) := \phi(\z)$, which is a linear ODE with boundary values $f_R(R) = 1$, $f_R'(R) = 0$. Thus, it has a unique solution.

For $R_0 < d(o,\z) < R$ we infer from \eqref{eq:euler lagrange} that $-\Delta \phi = 0$. As the Laplace-Beltrami operator on $\IH^d$ is given in hyperbolic polar coordinates by 
\[
\Delta = \sinh (r) ^{1-d} \partial_r \left( \sinh(t)^{d-1} \partial_r  \right)  + \sinh(r)^{-2} \Delta_\Sigma ,
\]
we find that $ \partial_r \left( \sinh(t)^{d-1} \partial_r f_R(r)  \right) = 0$. The corresponding solutions for $d=2$ and $d=3$ are given by \eqref{eq:harmonic solutions}.


For the energy we use partial integration and $f_\infty'(r) = \frac{1}{\sinh^{d-1} r}$. Thus we get
\begin{align*}
E_R &= \vol{\IS^{d-1}} \bigg( \mu \left[ \sinh^{d-1} r f_R(r) \partial_r f_R(r) \right]_0^R  \\ &\qquad +  \int_0^R  \left(-\mu \frac{1}{\sinh^{d-1}(r)} \partial_r ( \sinh^{d-1}(r) \partial_r f_R(r)  )  + \frac{1}{2} V(r) f_R(r) \right) f_R(r) \sinh^{d-1}( r ) \d r \bigg) \\
&= \frac{ \mu \vol{\IS^{d-1}}}{ f_\infty(R)^2 } \left[ \sinh^{d-1} r f_\infty(r) \partial_r f_\infty(r) \right]_0^R  \\ &\qquad + \vol{\IS^{d-1}}  \int_0^R  \underbrace{ \left(-\mu \Delta_r f_R(r)  + \frac{1}{2} V(r)f_R(r) \right) }_{=0} f_R(r)  \sinh^{d-1}(r) \d r \\
&=  \frac{ \mu \vol{\IS^{d-1}}  }{f_\infty(R)} \sinh^{d-1} (R)  f_\infty'(R).
\end{align*}

The last statement \eqref{eq:f_R greater} follows in the same way as in \cite[Lemma C.2]{greenbook} from the Hopf maximum principle. 
\end{proof}

\label{appendix}

\renewcommand*{\bibfont}{\footnotesize}
\printbibliography

\end{document}